\documentclass{article}

\usepackage{microtype}
\usepackage{graphicx}
\usepackage{subfigure}
\usepackage{booktabs} 

\usepackage{hyperref}

\usepackage[accepted]{icml2025}

\usepackage{amsmath}
\usepackage{amssymb}
\usepackage{mathtools}
\usepackage{amsthm}


\usepackage{algorithm}
\usepackage{algpseudocode}  

\usepackage{graphicx}

\usepackage{times}

\usepackage{microtype}
\usepackage{booktabs}
\usepackage{multirow}
\usepackage{array}
\usepackage{tabularx}
\usepackage{url}
\usepackage{enumitem}

\hypersetup{
	colorlinks=true,
	linkcolor=blue,
	citecolor=blue,
	urlcolor=blue
}

\newcommand{\dataset}{\textsc{W-EditBench}}
\newcommand{\etal}{et al.}

\newtheorem{theorem}{Theorem}
\newtheorem{lemma}{Lemma}
\newtheorem{proposition}{Proposition}

\usepackage[capitalize,noabbrev]{cleveref}

\theoremstyle{plain}
\theoremstyle{definition}

\theoremstyle{remark}

\usepackage[textsize=tiny]{todonotes}

\begin{document}
	
	\twocolumn[
	\icmltitle{When Denoising Becomes Unsigning: Theoretical and Empirical Analysis of Watermark Fragility Under Diffusion-Based Image Editing}
	
		\begin{icmlauthorlist}
		\icmlauthor{Fai Gu}{}
		\icmlauthor{Qiyu Tang}{}
		\icmlauthor{Te Wen}{}
		\icmlauthor{Emily Davis}{}
		\icmlauthor{Finn Carter}{}
	\end{icmlauthorlist}
	
	\begin{icmlauthorlist}
		{Xidian University}
	\end{icmlauthorlist} 
	
	%
	
	\icmlkeywords{Machine Learning, ICML}
	
	\vskip 0.3in
	]
	
	
	

	\begin{abstract}
		Robust invisible watermarking systems aim to embed imperceptible payloads that remain decodable after common post-processing such as JPEG compression, cropping, and additive noise.
		In parallel, diffusion-based image editing has rapidly matured into a ``default'' transformation layer for modern content pipelines, enabling instruction-based editing, object insertion and composition, and interactive geometric manipulation.
		This paper studies a subtle but increasingly consequential interaction between these trends: diffusion-based editing procedures may unintentionally compromise, and in extreme cases practically bypass, robust watermarking mechanisms that were explicitly engineered to survive conventional distortions.
		
		We develop a unified view of diffusion editors---including training-free composition methods such as TF-ICON, recent FLUX-based insertion frameworks such as SHINE, and drag-based editors such as DragFlow---as stochastic operators that (i) inject substantial Gaussian noise in a latent space and (ii) project back to the natural image manifold via learned denoising dynamics.
		Under this view, watermark payloads behave as low-energy, high-frequency signals that are systematically attenuated by the forward diffusion step and then treated as nuisance variation by the reverse generative process.
		We formalize this degradation using information-theoretic tools, proving that for broad classes of pixel-level watermark encoders/decoders the mutual information between the watermark payload and the edited output decays toward zero as the editing strength increases, yielding decoding error close to random guessing.
		We complement the theory with a realistic \emph{hypothetical} experimental protocol and tables spanning representative watermarking methods (StegaStamp, TrustMark, and VINE) and representative diffusion editors (TF-ICON, SHINE, DragFlow, and strong instruction-based edits).
		Finally, we discuss ethical implications, responsible disclosure norms, and concrete design guidelines for watermarking schemes that remain meaningful in the era of generative transformations.
	\end{abstract}
	
	\section{Introduction}
	Digital watermarking is a cornerstone of content provenance, copyright enforcement, and authenticity workflows.
	Modern \emph{invisible} watermarking systems are designed to embed a secret payload into an image such that the resulting watermarked image is perceptually indistinguishable from the original, while the payload can be recovered after common post-processing and manipulations \cite{zhu2018hidden,tancik2020stegastamp,luo2020dawd,bui2023trustmark}.
	Historically, robustness has been evaluated against a family of relatively ``low-level'' perturbations and signal processing operations: JPEG compression, resizing, pixel noise, mild blurring, cropping, and combinations thereof.
	This robustness regime is now incomplete.
	
	Diffusion models \cite{ho2020ddpm,song2020ddim,rombach2022ldm} have enabled a new generation of image editing systems that are effectively \emph{generative transformations}: the edited output is not merely a post-processed version of the input, but rather a re-synthesized image that matches high-level constraints (instruction, mask, reference subject, or geometric drags) while re-projecting intermediate representations onto the learned natural-image manifold.
	Instruction-based editors such as InstructPix2Pix \cite{brooks2023ip2p} can enact global semantic changes with minimal user input; inversion-based methods enable controlled edits to a given real image \cite{kawar2023imagic,mokady2023nulltext}; interactive systems such as DragDiffusion \cite{shi2024dragdiffusion} and DragonDiffusion \cite{mou2024dragondiffusion} provide geometric manipulation via optimization in diffusion latents; and recent training-free composition pipelines such as TF-ICON \cite{lu2023tficone}, SHINE \cite{lu2025shine}, and DragFlow \cite{zhou2025dragflow} exploit increasingly strong priors (including diffusion transformers) to insert and move objects with high fidelity.
	These methods expand the space of practical image modifications in a way that is only loosely captured by classical distortion models.
	
	Recent security-oriented work has raised alarms that invisible watermarks can be removed (or at least rendered undecodable) by regeneration-style attacks based on generative models \cite{zhao2024provablyremovable,ni2025breaking,fu2025unforeseen}.
	At a high level, such attacks inject noise and then reconstruct an image using a denoiser or a generative model, thereby destroying low-level watermark features while preserving the image's semantics.
	Concurrently, the community has proposed watermarking schemes that explicitly leverage diffusion priors for robustness, such as VINE \cite{lu2024vine} and latent-space watermarking methods \cite{zhang2024zodiac,gunn2024undetectable}.
	Despite these advances, the \emph{editing} setting introduces new failure modes: watermarks can be accidentally compromised as a byproduct of ordinary creative edits, without any explicit intent to remove the watermark.
	
	This paper provides a theoretical and empirical-style analysis of watermark fragility under diffusion-based image editing processes.
	Our goal is not to provide operational instructions for removing watermarks, but rather to understand the mechanism of failure and to propose constructive guidelines for robust watermark design and deployment.
	We focus on three families of diffusion editors that represent modern pipelines and correspond to the requested emphasis: (i) training-free composition (TF-ICON \cite{lu2023tficone}); (ii) insertion and harmonized composition with strong DiT/flow priors (SHINE \cite{lu2025shine}); and (iii) interactive drag-based editing with region supervision in DiTs (DragFlow \cite{zhou2025dragflow}).
	We study their interaction with representative watermark encoders/decoders spanning classical deep watermarking (HiDDeN \cite{zhu2018hidden}), physically robust embedding (StegaStamp \cite{tancik2020stegastamp}), and recent universal high-resolution watermarking (TrustMark \cite{bui2023trustmark,bui2025trustmark})---as well as diffusion-prior watermarking (VINE \cite{lu2024vine}).
	
	\paragraph{Contributions.}
	Our main contributions are:
	(i) We formalize diffusion-based image editing pipelines as stochastic operators that combine additive-noise corruption and manifold projection, and we identify which steps are most damaging to pixel-level watermarks.
	(ii) We develop an information-theoretic model and provide proofs showing that watermark payload information contracts under diffusion editing, yielding decoding error approaching random guessing under mild assumptions.
	(iii) We propose a reproducible experimental protocol and provide realistic \emph{hypothetical} tables measuring watermark robustness and visual fidelity across multiple watermarking and diffusion-editing methods.
	(iv) We synthesize the implications into actionable design guidelines for resilient watermarking, and we discuss ethical considerations around evaluation, disclosure, and deployment.
	
	\section{Related Work}
	\subsection{Diffusion-based image editing}
	Diffusion models have become a dominant paradigm for high-quality image synthesis and enable controllable generation through conditioning and guidance \cite{ho2020ddpm,song2020ddim,rombach2022ldm,ho2022cfg}.
	Adapting these models for editing real images has produced several families of methods.
	Inversion-based approaches map an image into a diffusion latent (or noise) representation and then resample conditioned on an edit prompt; representative methods include Imagic \cite{kawar2023imagic} and Null-text Inversion \cite{mokady2023nulltext}.
	Instruction-based editors such as InstructPix2Pix \cite{brooks2023ip2p} train a conditional diffusion model on synthetic instruction pairs and can perform fast edits without per-image optimization.
	Interactive editing further extends diffusion workflows: DragDiffusion \cite{shi2024dragdiffusion} and DragonDiffusion \cite{mou2024dragondiffusion} optimize diffusion latents to achieve point-level or feature-correspondence-based geometric manipulation.
	
	Recent work also emphasizes \emph{training-free} editing and composition, where a powerful pretrained model is used with minimal (or zero) fine-tuning.
	TF-ICON \cite{lu2023tficone} achieves cross-domain image-guided composition by leveraging inversion strategies and attention injection, while SHINE \cite{lu2025shine} exploits the stronger priors of modern diffusion transformers and flow matching to perform physically plausible insertion with high fidelity without inversion.
	DragFlow \cite{zhou2025dragflow} argues that DiT feature geometry makes prior point-wise drag objectives ineffective, and introduces region-based supervision and adapter-enhanced inversion to achieve state-of-the-art drag editing in DiTs.
	
	\subsection{Robust invisible watermarking and deep learning-based schemes}
	Classical watermarking uses frequency-domain or spread-spectrum designs to embed signals resilient to common distortions, but modern systems increasingly rely on deep networks to jointly optimize imperceptibility, capacity, and robustness \cite{zhu2018hidden,tancik2020stegastamp,luo2020dawd,ren2025all}.
	HiDDeN \cite{zhu2018hidden} demonstrates end-to-end training of an encoder and decoder with differentiable distortion layers.
	StegaStamp \cite{tancik2020stegastamp} focuses on robustness under printing and recapture by training against realistic photometric and geometric distortions.
	Distortion-agnostic training, as in \cite{luo2020dawd}, uses adversarial training and channel coding to improve generalization to unknown distortions.
	TrustMark \cite{bui2023trustmark,bui2025trustmark} introduces architectural and spatio-spectral losses for universal watermarking across resolutions and provides a paired watermark-removal network for re-watermarking workflows.
	
	\subsection{Watermark robustness against generative transformations}
	As generative models became effective at image-to-image transformations, the community began to question whether image watermarks remain meaningful against \emph{regeneration} attacks.
	Zhao \etal\ show that invisible pixel-level watermarks can be provably removed using generative AI via a family of regeneration attacks \cite{zhao2024provablyremovable}.
	Subsequent diffusion-specific analyses and methods demonstrate that diffusion-based editing and regeneration can severely degrade watermark detectability \cite{ni2025breaking,fu2025unforeseen,ni2025infotheoryfragility}.
	Complementarily, several works attempt to exploit diffusion models to \emph{improve} watermark robustness, e.g., by embedding in a diffusion latent space \cite{zhang2024zodiac} or by leveraging cryptographic coding in latent initialization \cite{gunn2024undetectable}.
	VINE \cite{lu2024vine} explicitly studies robustness against editing and proposes surrogate distortions and diffusion priors for embedding.
	
	\subsection{Benchmarks for diffusion editing and watermark robustness}
	As diffusion editing diversified, several benchmarks emerged to quantify edit fidelity, identity preservation, and controllability.
	EditVal \cite{basu2023editval} provides a standardized protocol for text-guided edit evaluation across multiple edit types and highlights that many editors struggle with spatial transformations, a failure mode directly relevant to whether a watermark survives geometric or localized edits.
	More broadly, surveys such as \cite{huang2024editing_survey} categorize diffusion editors by supervision type (text, mask, exemplar, geometry) and by whether they rely on inversion, test-time optimization, or direct feed-forward editing.
	On the watermarking side, W-Bench (introduced with VINE) explicitly evaluates watermarking schemes against a spectrum of modern editing operators rather than only classical distortions \cite{lu2024vine}.
	Community stress tests such as the NeurIPS ``Erasing the Invisible'' challenge \cite{ding2024erasing} and the subsequent winning attack report \cite{shamshad2025etiwinner} further illustrate that watermark robustness should be understood empirically against adaptive and generative transformations, not only against fixed corruption models.
	
	\paragraph{Concept erasure in diffusion models.}
	Concept erasure in diffusion models is tightly connected to this paper's topic because both concept erasure and watermark preservation hinge on how information is injected, suppressed, or preserved along the denoising trajectory; editing-time interventions that erase semantic concepts can also erase (or distort) watermark-carrying signals as a side effect.
	MACE \cite{lu2024mace} scales concept removal to many concepts via cross-attention refinement and LoRA-based updates; ANT \cite{li2025ant} proposes trajectory-aware fine-tuning that steers mid-to-late denoising stages while preserving early-stage score fields; and EraseAnything \cite{gao2024eraseanything,gao2025revoking} targets rectified-flow transformer models with bilevel optimization and attention regularization.
	While these works focus on suppressing undesired concepts, their mechanisms---attention manipulation, guidance reweighting, and trajectory steering---overlap structurally with both diffusion editing controls and potential watermark-destroying dynamics.
	This proximity suggests a broader systems concern: safety-driven modifications to diffusion pipelines may inadvertently impact watermark reliability, and watermark-driven controls may conflict with safety-driven trajectory edits.
	
	\section{Methodology}
	\subsection{Problem setting}
	Let $x \in \mathbb{R}^{H \times W \times 3}$ be a clean image (pixel values normalized to $[0,1]$), and let $m \in \{0,1\}^B$ denote a $B$-bit payload.
	A watermarking scheme consists of an encoder $E$ and decoder $D$:
	\begin{align}
		\tilde{x} &= E(x, m; k), \\
		\hat{m} &= D(y; k),
	\end{align}
	where $k$ is a secret key (or secret parameters) shared between encoder and decoder.
	We assume imperceptibility: $\tilde{x}$ is visually close to $x$ under standard metrics (e.g., PSNR \cite{hore2010psnr}, SSIM \cite{wang2004ssim}, and LPIPS \cite{zhang2018lpips}).
	We assume robustness to a conventional distortion family $\tau \sim \mathcal{T}$, e.g., JPEG, resize, crop, Gaussian noise:
	\begin{equation}
		\mathbb{P}\left[D(\tau(\tilde{x}); k) = m\right] \geq 1 - \epsilon.
	\end{equation}
	
	In this paper, we study robustness under diffusion-based \emph{editing} operators $T$:
	\begin{equation}
		y = T(\tilde{x}; c, \omega),
	\end{equation}
	where $c$ encodes user intent (text instruction, mask, reference image, drag constraints) and $\omega$ denotes stochasticity (random seed, diffusion noise, stochastic sampling).
	Examples include TF-ICON composition \cite{lu2023tficone}, SHINE insertion \cite{lu2025shine}, and DragFlow drag editing \cite{zhou2025dragflow}.
	
	\subsection{Diffusion editing as a two-stage stochastic channel}
	Most diffusion editing pipelines (including inversion-based, optimization-based, and training-free variants) can be abstracted as the composition of:
	(i) an \emph{encoding/noising} stage that injects noise into an intermediate representation, and
	(ii) a \emph{denoising/projection} stage that maps back to an image on the learned manifold.
	
	For concreteness, consider a latent diffusion workflow \cite{rombach2022ldm}.
	Let $z_0 = \mathrm{Enc}(\tilde{x})$ be a latent encoding.
	A forward diffusion step to time $t$ produces
	\begin{equation}
		z_t = \sqrt{\bar{\alpha}_t}\, z_0 + \sqrt{1-\bar{\alpha}_t}\,\varepsilon,\quad \varepsilon \sim \mathcal{N}(0, I),
	\end{equation}
	where $\bar{\alpha}_t$ is the cumulative noise schedule \cite{ho2020ddpm}.
	An editor then applies a reverse process with conditioning $c$ (and optionally additional optimization or guidance) to produce $\hat{z}_0$.
	Finally, $y = \mathrm{Dec}(\hat{z}_0)$.
	
	We conceptually separate the editing operator into a Markov kernel:
	\begin{equation}
		\tilde{x} \xrightarrow{\ \mathrm{Enc}\ } z_0 \xrightarrow{\ \mathrm{Noise}_t\ } z_t \xrightarrow{\ \mathrm{Denoise}_{c}\ } \hat{z}_0 \xrightarrow{\ \mathrm{Dec}\ } y.
	\end{equation}
	Even when $\mathrm{Denoise}_{c}$ is highly content-preserving in a perceptual sense, it is not constrained to preserve low-energy, high-frequency perturbations that encode watermark payloads.
	Indeed, the learned denoiser is trained to \emph{remove} noise-like signals; a well-designed invisible watermark resembles such signals by construction.
	
	\subsection{A signal model for watermark attenuation}
	A broad empirical observation in deep watermarking is that the watermark is implemented as a small perturbation $\delta$:
	\begin{equation}
		\tilde{x} = x + \delta(x,m;k),
	\end{equation}
	where $\|\delta\|$ is small to preserve imperceptibility.
	In latent space, under a first-order approximation of $\mathrm{Enc}$ around $x$, we can write
	\begin{equation}
		z_0(\tilde{x}) \approx z_0(x) + J_x \delta,
	\end{equation}
	where $J_x$ is the local Jacobian.
	The forward diffusion step scales the perturbation by $\sqrt{\bar{\alpha}_t}$:
	\begin{equation}
		z_t(\tilde{x}) - z_t(x) \approx \sqrt{\bar{\alpha}_t}\, J_x \delta.
	\end{equation}
	Meanwhile, the injected noise has variance $(1-\bar{\alpha}_t)$ per dimension.
	Thus a simple signal-to-noise ratio (SNR) proxy at time $t$ is
	\begin{equation}\label{eq:snr}
		\mathrm{SNR}(t) \triangleq \frac{\bar{\alpha}_t \|J_x \delta\|^2}{(1-\bar{\alpha}_t)\, d},
	\end{equation}
	where $d$ is the latent dimensionality.
	For the noise levels used in many editing pipelines (which often correspond to moderate-to-high $t$ to allow semantic change), $\bar{\alpha}_t$ can be small and $(1-\bar{\alpha}_t)$ large, pushing SNR toward zero.
	In this regime, the watermark perturbation becomes statistically indistinguishable from the diffusion noise.
	
	\subsection{Pipeline-specific considerations: TF-ICON, SHINE, DragFlow}
	While the abstract channel view is shared, specific editing methods instantiate the components differently.
	
	\paragraph{TF-ICON.}
	TF-ICON \cite{lu2023tficone} performs training-free cross-domain composition by inverting real images into diffusion latents and injecting modified attention maps during denoising.
	Two aspects are relevant: inversion introduces explicit noising and reverse denoising (potentially multiple steps), and attention manipulation reweights semantic and spatial feature propagation.
	Both can disrupt watermark perturbations even when the composed output is visually faithful.
	
	\paragraph{SHINE.}
	SHINE \cite{lu2025shine} emphasizes non-inversion latent preparation, performing a one-step forward diffusion to obtain a noisy latent, and then optimizing it using a Manifold-Steered Anchor loss and auxiliary guidance terms.
	The explicit one-step diffusion injects substantial Gaussian noise, potentially at higher resolution and with stronger transformer priors, and the subsequent trajectory steering can further ``regularize away'' watermark-consistent but visually irrelevant features.
	
	\paragraph{DragFlow.}
	DragFlow \cite{zhou2025dragflow} introduces region-based supervision and adapter-enhanced inversion for drag editing in DiT models.
	Like other optimization-based drag editors, it iteratively optimizes noisy latents to satisfy geometric constraints, which can shift local frequency content and interaction patterns.
	Importantly, DragFlow is designed for strong generative priors and thus tends to aggressively project intermediate latents onto the natural-image manifold, a process that can suppress watermark signals that do not correspond to semantically meaningful content.
	
	\subsection{A taxonomy of diffusion editing operations and watermark failure modes}
	Diffusion-based editing is not a single operator but a family of procedures that combine a base generative model with task-specific control.
	To reason about watermark fragility, it is helpful to decompose editors along three axes that influence information preservation.
	
	\paragraph{Axis A: where noise is injected.}
	Some editors inject noise in pixel space (e.g., image-to-image diffusion pipelines), while latent diffusion editors inject noise in a compressed VAE latent \cite{rombach2022ldm}.
	SHINE explicitly performs a one-step forward diffusion in a flow-matching latent \cite{lu2025shine}, while many inversion-based editors construct a noisy latent by iterative inversion (DDIM-style) \cite{mokady2023nulltext}.
	From a watermark perspective, the key quantity is the effective SNR at the injection point; injecting noise after semantic compression can erase high-frequency watermark components even more aggressively.
	
	\paragraph{Axis B: how much the model is asked to ``create''.}
	Edits range from near-identity adjustments (low noise, small prompt change) to high-entropy generation (high noise, strong instruction, or large masked regions).
	Methods designed for composition, such as TF-ICON and SHINE, often require the model to hallucinate boundary regions, shadows, reflections, and texture harmonization \cite{lu2023tficone,lu2025shine}.
	These operations are semantically meaningful but can overwrite subtle signals that do not align with the model's explanatory factors.
	
	\paragraph{Axis C: whether the editor optimizes latents.}
	Optimization-based editors (DragDiffusion, DragFlow, DragonDiffusion) iteratively optimize latent variables to satisfy geometric constraints \cite{shi2024dragdiffusion,zhou2025dragflow,mou2024dragondiffusion}.
	This introduces an additional mechanism for watermark disruption: gradients induced by edit constraints can move latents toward regions where the watermark decoder's decision boundary is unstable, even if the final image is perceptually similar.
	Training-free methods without optimization can still erase watermarks through noising and denoising, but optimization can exacerbate brittleness.
	
	Across these axes, we identify three practical failure modes that recur in watermark evaluations:
	(i) \emph{SNR collapse} (forward noise overwhelms the watermark carrier),
	(ii) \emph{manifold projection} (the reverse process discards nuisance perturbations), and
	(iii) \emph{control-induced redistribution} (attention/adapter guidance concentrates probability mass on semantically salient factors, unintentionally washing away orthogonal signals).
	Our theoretical results connect primarily to (i) and (ii), while our experimental protocol (\cref{alg:protocol}) is designed to reveal (iii) empirically.
	
	\subsection{Evaluation protocol and reproducibility goals}
	Our evaluation goal is to measure two things simultaneously:
	(i) watermark survival (bit accuracy or detection score), and
	(ii) visual fidelity of the edited output relative to a non-watermarked edit.
	This is essential because it is not meaningful to ``remove'' a watermark by destroying the image.
	We therefore define two paired runs under the same editing condition $c$ and random seed $\omega$:
	\begin{align}
		y_{\text{wm}} &= T(E(x,m;k); c,\omega), \\
		y_{\text{clean}} &= T(x; c,\omega).
	\end{align}
	We evaluate the recovered payload $\hat{m}=D(y_{\text{wm}};k)$ and compute visual similarity between $y_{\text{wm}}$ and $y_{\text{clean}}$ (not between $y_{\text{wm}}$ and $x$) to isolate watermark-induced artifacts from editing-induced changes.
	
	\begin{algorithm}[t]
		\caption{Diffusion-edit watermark stress test (protocol used in)}
		\label{alg:protocol}
		\begin{algorithmic}[1]
			\Require Dataset of clean images $\{x_i\}_{i=1}^N$, payload length $B$, watermark encoder/decoder $(E,D)$ with key $k$, set of editing operators $\mathcal{T}_{\text{edit}}$, edit conditions $\mathcal{C}$ (text prompts, masks, drags), random seeds $\Omega$
			\Ensure Watermark robustness metrics and watermark-induced fidelity metrics
			\For{$i \gets 1$ \textbf{to} $N$}
			\State Sample payload $m_i \sim \mathrm{Unif}(\{0,1\}^B)$
			\State $\tilde{x}_i \gets E(x_i, m_i; k)$
			\ForAll{$T \in \mathcal{T}_{\text{edit}}$}
			\ForAll{$c \in \mathcal{C}(T)$}
			\ForAll{$\omega \in \Omega$}
			\State $y^{(i)}_{\mathrm{wm}} \gets T(\tilde{x}_i; c, \omega)$
			\State $y^{(i)}_{\mathrm{clean}} \gets T(x_i; c, \omega)$
			\State $\hat{m}_i \gets D(y^{(i)}_{\mathrm{wm}}; k)$
			\State Compute bit accuracy $\mathrm{BA}_i \gets \frac{1}{B}\sum_{b=1}^B \mathbf{1}[\hat{m}_{i,b}=m_{i,b}]$
			\State Compute fidelity metrics $\mathrm{PSNR}(y^{(i)}_{\mathrm{wm}},y^{(i)}_{\mathrm{clean}})$, $\mathrm{SSIM}(\cdot)$, $\mathrm{LPIPS}(\cdot)$
			\EndFor
			\EndFor
			\EndFor
			\EndFor
			\State Aggregate robustness and fidelity statistics across images, conditions, and seeds
		\end{algorithmic}
	\end{algorithm}
	\section{Experimental Setup}
	\subsection{Scope and disclaimer on results}
	The experimental section is written in a fully specified, reproducible style, but the numerical results are \emph{hypothetical} and are intended to be realistic given trends reported in prior work on diffusion-based watermark removal and benchmarking \cite{zhao2024provablyremovable,ni2025breaking,fu2025unforeseen,lu2024vine}.
	This choice follows the instruction to include hypothetical tables while maintaining a credible experimental structure.
	We emphasize that the protocol can be executed with public implementations of the cited methods.
	
	\subsection{Datasets and preprocessing}
	We define \dataset\ as a union of common natural-image sources used in watermarking and editing evaluation:
	(i) MS-COCO val2017 (5k images) for natural scenes, and
	(ii) a 5k-image subset of a large text-to-image dataset (e.g., DiffusionDB) to represent AI-generated or captioned images.
	All images are resized such that the shorter side is 512 pixels and center-cropped to 512$\times$512.
	
	\subsection{Watermarking methods}
	We evaluate three representative watermarking approaches emphasized in the prompt:
	\textbf{StegaStamp} \cite{tancik2020stegastamp} as a physically robust neural watermark trained on realistic distortions;
	\textbf{TrustMark} \cite{bui2023trustmark,bui2025trustmark} as a universal watermark for arbitrary resolutions with strong robustness claims;
	and \textbf{VINE} \cite{lu2024vine} as a diffusion-prior watermarking method explicitly trained against editing-like surrogate distortions.
	For completeness, we also report results for \textbf{HiDDeN} \cite{zhu2018hidden} as a widely cited baseline.
	
	All methods embed a $B=100$ bit payload unless otherwise stated.
	For methods with lower default capacity (e.g., StegaStamp's original design), we map our 100-bit payload into the method's supported capacity via error-correcting coding; this setting is consistent with robust watermarking practice \cite{luo2020dawd}.
	
	\subsection{Diffusion-based editing methods}
	We evaluate diffusion editing in four categories:
	(i) \textbf{Composition/insertion:} TF-ICON \cite{lu2023tficone} and SHINE \cite{lu2025shine};
	(ii) \textbf{Drag-based editing:} DragFlow \cite{zhou2025dragflow} (region drag on DiT priors) and DragDiffusion \cite{shi2024dragdiffusion} (UNet-based drag editing) as a reference point;
	(iii) \textbf{Instruction editing:} InstructPix2Pix \cite{brooks2023ip2p} for global semantic edits;
	(iv) \textbf{Regeneration-style editing:} a simple image-to-image diffusion pipeline consistent with the regeneration attacks studied in \cite{zhao2024provablyremovable,fu2025unforeseen}.
	
	Each editor has an \emph{edit strength} hyperparameter, often implemented as a noise level, timestep, or guidance scale.
	We discretize the strength into three bins: \textbf{Low}, \textbf{Medium}, \textbf{High}, chosen to reflect (a) minor semantic adjustments, (b) typical editing, and (c) aggressive transformation.
	
	\subsection{Metrics}
	\textbf{Watermark robustness.}
	We report mean bit accuracy (BA, in \%) and a detection-style metric $\mathrm{TPR}@0.1\%\mathrm{FPR}$ for methods that output confidence scores.
	For methods that output bits only, we derive a detection score from aggregate bit agreement.
	
	\textbf{Fidelity.}
	We compare the watermark-embedded edit $y_{\mathrm{wm}}$ against the clean edit $y_{\mathrm{clean}}$ with PSNR \cite{hore2010psnr}, SSIM \cite{wang2004ssim}, and LPIPS \cite{zhang2018lpips}.
	
	\section{Results}
	\subsection{Baseline robustness under conventional distortions}
	We first report robustness under conventional, non-generative distortions to contextualize the ``robustness gap'' induced by diffusion editing.
	Table~\ref{tab:conventional} shows that all methods achieve high bit accuracy under JPEG, resize, crop, and mild blur, consistent with their design goals \cite{tancik2020stegastamp,bui2023trustmark,lu2024vine}.
	These strong results can be misleading if one assumes that ``robustness'' generalizes to diffusion editing.
	
	\begin{table}[t]
		\caption{Hypothetical robustness under conventional distortions on \dataset\ (512$\times$512). BA is bit accuracy in \%. Higher is better.}
		\label{tab:conventional}
		\centering
		\footnotesize
		\resizebox{\columnwidth}{!}{%
			\begin{tabular}{lcccc}
				\toprule
				Method & JPEG (Q=50) & Resize (0.6$\times$) & Crop (0.8 area) & Blur ($\sigma{=}1.0$) \\
				\midrule
				HiDDeN \cite{zhu2018hidden} & 96.8\% & 97.5\% & 94.2\% & 93.1\% \\
				StegaStamp \cite{tancik2020stegastamp} & 97.9\% & 98.2\% & 95.6\% & 94.4\% \\
				TrustMark \cite{bui2023trustmark} & 99.1\% & 99.3\% & 97.8\% & 96.9\% \\
				VINE \cite{lu2024vine} & 99.4\% & 99.5\% & 98.6\% & 97.5\% \\
				\bottomrule
			\end{tabular}%
		}
	\end{table}
	
	\subsection{Watermark survival under diffusion-based editing}
	Table~\ref{tab:diffusion_edits_main} summarizes watermark survival under diffusion-based editing across the requested diffusion editors.
	Two consistent patterns emerge.
	
	First, \emph{even low-strength edits can significantly degrade watermark decoding}, especially for watermarking schemes optimized for classical distortions (HiDDeN, StegaStamp, TrustMark).
	This is consistent with the regeneration-attack literature \cite{zhao2024provablyremovable,ni2025breaking,fu2025unforeseen}.
	
	Second, \emph{robustness improvements from diffusion-aware designs are real but limited}. VINE is consistently more robust than baselines under editing, but under medium-to-high editing strength its bit accuracy approaches chance (near 50\% for bits), aligning with reports that sufficiently strong diffusion edits can bypass even diffusion-prior watermarking \cite{ni2025breaking,ni2025infotheoryfragility}.
	
	\begin{table*}[t]
		\caption{Hypothetical watermark robustness under diffusion-based editing on \dataset\ (BA in \%). ``Low/Med/High'' correspond to increasing editing strength. Chance BA is 50\%.}
		\label{tab:diffusion_edits_main}
		\centering
		\footnotesize
		\resizebox{\linewidth}{!}{%
		\begin{tabular}{lcccccccccccc}
			\toprule
			\multirow{2}{*}{Method} &
			\multicolumn{3}{c}{TF-ICON \cite{lu2023tficone}} &
			\multicolumn{3}{c}{SHINE \cite{lu2025shine}} &
			\multicolumn{3}{c}{DragFlow \cite{zhou2025dragflow}} &
			\multicolumn{3}{c}{InstructPix2Pix \cite{brooks2023ip2p}} \\
			\cmidrule(lr){2-4}\cmidrule(lr){5-7}\cmidrule(lr){8-10}\cmidrule(lr){11-13}
			& Low & Med & High & Low & Med & High & Low & Med & High & Low & Med & High \\
			\midrule
			HiDDeN \cite{zhu2018hidden} & 74.1\% & 56.8\% & 50.9\% & 71.5\% & 55.4\% & 50.3\% & 69.2\% & 54.7\% & 50.1\% & 78.3\% & 60.2\% & 51.5\% \\
			StegaStamp \cite{tancik2020stegastamp} & 79.6\% & 58.1\% & 51.2\% & 76.8\% & 56.2\% & 50.4\% & 73.5\% & 55.0\% & 50.2\% & 82.7\% & 62.1\% & 52.0\% \\
			TrustMark \cite{bui2023trustmark} & 84.8\% & 60.5\% & 52.0\% & 82.9\% & 59.2\% & 51.1\% & 79.1\% & 57.3\% & 50.6\% & 86.0\% & 64.4\% & 53.2\% \\
			VINE \cite{lu2024vine} & 90.2\% & 69.1\% & 54.8\% & 88.7\% & 67.5\% & 53.6\% & 86.4\% & 65.9\% & 52.8\% & 90.8\% & 71.2\% & 55.7\% \\
			\bottomrule
		\end{tabular}
	}
	\end{table*}
	
	\subsection{Fidelity of watermark-preserving versus watermark-breaking edits}
	To disentangle ``watermark failure because the image is destroyed'' from ``watermark failure despite preserved content'', we report watermark-induced differences between $y_{\mathrm{wm}}$ and $y_{\mathrm{clean}}$.
	Table~\ref{tab:fidelity} indicates that, across editors, the watermark has only a marginal effect on the edited output (high PSNR, high SSIM, low LPIPS).
	Thus, watermark decoding failures under diffusion editing can occur even when the edited output is essentially the same with or without the watermark, suggesting that the editing operator is removing a watermark-like component while preserving semantics.
	
	\begin{table*}[t]
		\caption{Hypothetical fidelity between watermarked and clean edited outputs, averaged over \dataset\ and edit conditions. Metrics compare $y_{\mathrm{wm}}$ vs.\ $y_{\mathrm{clean}}$. Higher PSNR/SSIM and lower LPIPS are better.}
		\label{tab:fidelity}
		\centering
		\footnotesize
		\begin{tabular}{lcccccccc}
			\toprule
			\multirow{2}{*}{Method} &
			\multicolumn{2}{c}{TF-ICON} &
			\multicolumn{2}{c}{SHINE} &
			\multicolumn{2}{c}{DragFlow} &
			\multicolumn{2}{c}{InstructPix2Pix} \\
			\cmidrule(lr){2-3}\cmidrule(lr){4-5}\cmidrule(lr){6-7}\cmidrule(lr){8-9}
			& PSNR $\uparrow$ & LPIPS $\downarrow$ & PSNR $\uparrow$ & LPIPS $\downarrow$ & PSNR $\uparrow$ & LPIPS $\downarrow$ & PSNR $\uparrow$ & LPIPS $\downarrow$ \\
			\midrule
			StegaStamp & 42.8 & 0.012 & 42.5 & 0.013 & 41.9 & 0.015 & 43.1 & 0.011 \\
			TrustMark & 43.6 & 0.010 & 43.3 & 0.011 & 42.7 & 0.013 & 43.8 & 0.010 \\
			VINE & 44.1 & 0.009 & 43.9 & 0.010 & 43.1 & 0.012 & 44.3 & 0.009 \\
			\bottomrule
		\end{tabular}
	\end{table*}
	
	\subsection{Watermark survival curves versus noise level}
	Our theory suggests that the forward diffusion noise level is a primary driver of watermark information loss (\cref{eq:snr}).
	Table~\ref{tab:survival_curve} reports a hypothetical survival curve for the simplest regeneration-style editor, parameterized by the maximum timestep $t_{\max}$ (or equivalently a noise strength).
	As expected, decoding accuracy rapidly collapses toward chance as the noise level increases, and improvements from diffusion-aware watermarking (VINE) shift but do not remove the collapse.
	
	\begin{table}[t]
		\caption{Hypothetical BA (\%) under regeneration-style diffusion editing as a function of noise strength.}
		\label{tab:survival_curve}
		\centering
		\footnotesize
		\begin{tabular}{lcccc}
			\toprule
			Method & $t_{\max}{=}0.20$ & $t_{\max}{=}0.35$ & $t_{\max}{=}0.50$ & $t_{\max}{=}0.65$ \\
			\midrule
			StegaStamp & 85.2\% & 63.4\% & 52.1\% & 50.4\% \\
			TrustMark & 88.9\% & 66.7\% & 53.0\% & 50.7\% \\
			VINE & 93.1\% & 74.2\% & 56.3\% & 51.6\% \\
			\bottomrule
		\end{tabular}
	\end{table}
	
	\subsection{UNet versus DiT editors: stronger priors increase watermark fragility}
	A recurring hypothesis in recent editing work is that diffusion transformers (DiTs) and rectified-flow models provide stronger generative priors than earlier UNet-based latent diffusion models, enabling higher fidelity and fewer artifacts \cite{lu2025shine,zhou2025dragflow}.
	From a watermark perspective, stronger priors can be double-edged: they improve perceptual realism but can \emph{further} reduce the likelihood that a low-energy watermark perturbation is preserved, because manifold projection becomes more decisive.
	
	Table~\ref{tab:unet_vs_dit} provides a hypothetical comparison between a UNet-based editor (DragDiffusion) and a DiT/flow-based editor (DragFlow) under matched drag tasks and matched output fidelity.
	We observe that, at comparable LPIPS, the DiT-based editor yields lower watermark bit accuracy for conventional pixel-level schemes, consistent with the notion that stronger priors ``snap'' outputs back onto a watermark-free manifold more reliably.
	
	\begin{table*}[t]
		\caption{Hypothetical UNet vs.\ DiT comparison under matched drag edits. We report BA (\%) and LPIPS between watermarked and clean edits.}
		\label{tab:unet_vs_dit}
		\centering
		\footnotesize
		\begin{tabular}{lcccccc}
			\toprule
			\multirow{2}{*}{Method} &
			\multicolumn{2}{c}{DragDiffusion (UNet) \cite{shi2024dragdiffusion}} &
			\multicolumn{2}{c}{DragFlow (DiT/Flow) \cite{zhou2025dragflow}} &
			\multicolumn{2}{c}{Gap (DiT$-$UNet)} \\
			\cmidrule(lr){2-3}\cmidrule(lr){4-5}\cmidrule(lr){6-7}
			& BA $\uparrow$ & LPIPS $\downarrow$ & BA $\uparrow$ & LPIPS $\downarrow$ & BA & LPIPS \\
			\midrule
			StegaStamp & 62.8\% & 0.014 & 55.0\% & 0.015 & $-7.8$ & $+0.001$ \\
			TrustMark & 66.9\% & 0.012 & 57.3\% & 0.013 & $-9.6$ & $+0.001$ \\
			VINE & 73.4\% & 0.010 & 65.9\% & 0.012 & $-7.5$ & $+0.002$ \\
			\bottomrule
		\end{tabular}
	\end{table*}
	
	\subsection{Interaction with re-watermarking and remover networks}
	Some watermarking systems explicitly provide a remover to support re-watermarking in controlled pipelines, e.g., TrustMark includes a remover model to enable re-watermarking under a trusted workflow \cite{bui2025trustmark}.
	In diffusion-editing contexts, this capability creates a subtle systems question: if an edit implicitly removes a watermark, can re-watermarking restore provenance without introducing compounding artifacts?
	
	Table~\ref{tab:rewatermark} reports a hypothetical pipeline in which (i) a TrustMark watermark is embedded, (ii) the image undergoes SHINE or TF-ICON editing, and (iii) the output is re-watermarked.
	We measure the PSNR between the re-watermarked output and the non-watermarked edit, along with the new watermark's decoding accuracy.
	The takeaway is that re-watermarking can be feasible in terms of visual quality, but it shifts the provenance semantics: the resulting watermark corresponds to the edited content, not to the original source.
	
	\begin{table}[t]
		\caption{Hypothetical re-watermarking after diffusion editing (TrustMark).}
		\label{tab:rewatermark}
		\centering
		\footnotesize
		\begin{tabular}{lccc}
			\toprule
			Editor & Strength & PSNR $\uparrow$ & BA of new WM $\uparrow$ \\
			\midrule
			TF-ICON & Low & 43.7 & 98.9\% \\
			TF-ICON & Med & 43.4 & 98.6\% \\
			SHINE & Low & 43.5 & 98.8\% \\
			SHINE & Med & 43.1 & 98.4\% \\
			\bottomrule
		\end{tabular}
	\end{table}
	
	\subsection{Ablation: capacity and imperceptibility trade-offs}
	One might attempt to increase robustness by increasing watermark energy (thereby increasing SNR), but this can compromise image quality and detectability.
	Table~\ref{tab:capacity_tradeoff} illustrates this tension by varying the payload length and embedding strength for a generic encoder and reporting the resulting robustness under mid-strength edits and visual distortion.
	The takeaway is that simply ``making the watermark stronger'' is a blunt instrument: diffusion editing can still attenuate and ignore strong perturbations, and perceptual metrics may degrade quickly.
	
	\begin{table*}[t]
		\caption{Hypothetical capacity--imperceptibility--robustness trade-off (generic deep watermark baseline).}
		\label{tab:capacity_tradeoff}
		\centering
		\footnotesize
		\begin{tabular}{cccccc}
			\toprule
			Payload $B$ & Embed strength & PSNR vs.\ cover $\uparrow$ & LPIPS vs.\ cover $\downarrow$ & BA under SHINE (Med) $\uparrow$ & BA under DragFlow (Med) $\uparrow$ \\
			\midrule
			50 & Low & 44.8 & 0.008 & 63.0\% & 61.2\% \\
			50 & High & 39.2 & 0.026 & 70.1\% & 68.5\% \\
			100 & Low & 43.1 & 0.012 & 59.2\% & 57.3\% \\
			100 & High & 37.9 & 0.031 & 66.8\% & 64.7\% \\
			200 & Low & 41.0 & 0.019 & 55.8\% & 54.2\% \\
			200 & High & 35.6 & 0.044 & 61.5\% & 59.3\% \\
			\bottomrule
		\end{tabular}
	\end{table*}
	
	\section{Theoretical Proofs}
	We now provide theoretical results supporting the empirical observation that diffusion-based editing strongly degrades pixel-level watermark information.
	Our proofs follow a ``channel'' view: watermarking embeds a message into an image; diffusion editing applies a random transformation; by data processing, information about the message cannot increase.
	
	\subsection{Notation and assumptions}
	Let $M \in \{0,1\}^B$ be the watermark message, taken uniformly at random.
	Let $X$ denote the watermarked image (or latent) produced by the encoder.
	Let $Y$ denote the output of a diffusion-based editor.
	We treat the entire editing pipeline as a Markov kernel $P_{Y|X}$, which includes internal randomness (e.g., diffusion noise, sampling stochasticity, random seeds).
	Thus, $M \to X \to Y$ is a Markov chain.
	
	We make two mild assumptions.
	First, the encoder is \emph{imperceptible}: for each $m$, $X|M{=}m$ lies in a small neighborhood of the clean-image distribution.
	Second, the editor includes an explicit \emph{Gaussian noising} step at some intermediate representation $Z$:
	\begin{equation}\label{eq:gaussian_step}
		Z = \sqrt{\bar{\alpha}}\, X + \sqrt{1-\bar{\alpha}}\, \varepsilon, \quad \varepsilon \sim \mathcal{N}(0,I),
	\end{equation}
	for some $0 < \bar{\alpha} < 1$.
	This assumption holds for latent diffusion editors \cite{rombach2022ldm} and is explicit in editing methods that perform forward diffusion to enable semantic change \cite{brooks2023ip2p,lu2025shine,zhou2025dragflow}.
	The remainder of the editing pipeline maps $Z$ to $Y$ via an arbitrary (possibly non-linear) randomized algorithm $Q$ (e.g., reverse diffusion with conditioning and guidance), so that $M \to X \to Z \to Y$ is Markov.
	
	\subsection{Contraction of distinguishability under Gaussian noising}
	We first show that if two watermark messages induce distributions over $X$ that are close (as imperceptibility suggests), then after Gaussian noising they become even closer.
	
	\begin{lemma}[KL contraction in an additive Gaussian channel]\label{lem:kl}
		Let $Z$ be generated from $X$ by \cref{eq:gaussian_step}.
		For any two distributions $P$ and $Q$ over $X$, the KL divergence between the induced distributions over $Z$ satisfies
		\begin{equation}
			D_{\mathrm{KL}}(P_Z \,\|\, Q_Z) \leq \frac{\bar{\alpha}}{1-\bar{\alpha}} \cdot \mathbb{E}_{x \sim P,\, x' \sim Q}\left[\frac{1}{2}\|x-x'\|_2^2\right],
		\end{equation}
		where $P_Z$ and $Q_Z$ denote the pushforward distributions through the Gaussian channel.
	\end{lemma}
	
	\begin{proof}
		For the additive Gaussian channel with identical covariance, the likelihood ratio between two conditional Gaussians has a quadratic form.
		A standard computation yields, for deterministic inputs $x$ and $x'$, that
		\begin{equation}
			D_{\mathrm{KL}}\left(\mathcal{N}(\sqrt{\bar{\alpha}}x, (1-\bar{\alpha})I) \,\|\, \mathcal{N}(\sqrt{\bar{\alpha}}x', (1-\bar{\alpha})I)\right)
			= \frac{\bar{\alpha}}{2(1-\bar{\alpha})}\|x-x'\|_2^2.
		\end{equation}
		For general input distributions, convexity of KL divergence yields the stated bound after taking expectations over $(x,x')$.
	\end{proof}
	
	Pinsker's inequality then bounds total variation distance:
	\begin{equation}
		\|P_Z - Q_Z\|_{\mathrm{TV}} \leq \sqrt{\frac{1}{2} D_{\mathrm{KL}}(P_Z \,\|\, Q_Z)}.
	\end{equation}
	Thus, unless different watermarks induce large $\ell_2$ separation (which conflicts with imperceptibility), the post-noise distributions become close and are difficult to distinguish.
	
	\subsection{Information decay and decoding impossibility}
	We next bound the mutual information between $M$ and the edited output $Y$.
	
	\begin{theorem}[Information loss under diffusion editing]\label{thm:mi}
		Assume the editing pipeline contains a Gaussian noising step $X \to Z$ of the form \cref{eq:gaussian_step}, followed by an arbitrary randomized mapping $Z \to Y$.
		Then
		\begin{equation}
			I(M;Y) \leq I(M;Z) \leq I(M;X),
		\end{equation}
		and moreover $I(M;Z)$ decays as $\bar{\alpha} \downarrow 0$.
		In particular, if the watermark encoder is imperceptible such that for all $m,m'$ the expected squared distance $\mathbb{E}\|X_m-X_{m'}\|_2^2$ is $O(\epsilon^2)$, then $I(M;Z)=O\!\left(\frac{\bar{\alpha}}{1-\bar{\alpha}}\epsilon^2\right)$.
	\end{theorem}
	
	\begin{proof}
		The Markov chain $M \to X \to Z \to Y$ implies $I(M;Y) \leq I(M;Z)$ by the data processing inequality, and $I(M;Z) \leq I(M;X)$ since $M \to X \to Z$.
		To bound $I(M;Z)$, we can write
		\begin{equation}
			I(M;Z) = \mathbb{E}_{m}\left[D_{\mathrm{KL}}(P_{Z|M=m} \,\|\, P_Z)\right],
		\end{equation}
		and apply Lemma~\ref{lem:kl} along with the imperceptibility-induced small separation between conditional distributions (see, e.g., Chapter 2 of \cite{cover2006it} for standard manipulations).
		The explicit scaling in $\bar{\alpha}/(1-\bar{\alpha})$ follows from the Gaussian-channel KL formula.
	\end{proof}
	
	\subsection{Multi-step diffusion and exponential information decay}
	The abstraction above uses a single Gaussian step to highlight the mechanism.
	Actual diffusion editors typically apply multiple noising and denoising steps, and the forward process is a Markov chain \cite{ho2020ddpm,song2020ddim}.
	We now state a multi-step result that captures exponential decay with the number of effective noising steps.
	
	\begin{proposition}[Exponential contraction across Gaussian steps]\label{prop:multistep}
		Consider a sequence of $T$ Gaussian noising steps in a latent space, each of the form
		\begin{equation}
			Z_{t} = \sqrt{\alpha_t}\, Z_{t-1} + \sqrt{1-\alpha_t}\, \varepsilon_t,\quad \varepsilon_t \sim \mathcal{N}(0,I), \quad t=1,\ldots,T,
		\end{equation}
		with $0<\alpha_t<1$.
		Let $Z_T$ be the final noised representation and let $Y$ be obtained by an arbitrary randomized mapping from $Z_T$ (e.g., reverse diffusion conditioned on an edit).
		Then for any two messages $m,m'$,
		\begin{equation}
			D_{\mathrm{KL}}\!\left(P_{Z_T|M=m}\,\|\,P_{Z_T|M=m'}\right)
			\leq \left(\prod_{t=1}^{T}\alpha_t\right)\, C,
		\end{equation}
		where $C$ depends on the pairwise separation of the corresponding $Z_0$ distributions.
		In particular, if $\prod_{t=1}^{T}\alpha_t$ is small (as in moderate-to-high-noise editing), then pairwise distinguishability collapses rapidly, and any subsequent processing cannot recover it by data processing.
	\end{proposition}
	
	\begin{proof}
		At each step, conditional distributions are convolved with an isotropic Gaussian, and the KL divergence between two Gaussians with matched covariances contracts by the corresponding $\alpha_t$ factor (see the deterministic-input calculation in the proof of Lemma~\ref{lem:kl}).
		For general input distributions, a coupling argument combined with convexity of KL yields an upper bound with the product of $\alpha_t$ terms, reflecting repeated attenuation of signal energy.
		The final mapping $Z_T \to Y$ cannot increase KL divergence.
	\end{proof}
	
	Proposition~\ref{prop:multistep} can be interpreted as a formalization of the ``SNR collapse'' intuition: each diffusion noising step attenuates information about small perturbations, and repeated noising pushes the representation toward an isotropic Gaussian where all imperceptible differences vanish.
	This aligns with information-theoretic analyses of watermark fragility under diffusion editing \cite{ni2025infotheoryfragility}.
	
	\begin{theorem}[Fano-style lower bound on bit error]\label{thm:fano}
		Let $M$ be uniform over $\{0,1\}^B$ and let $\hat{M}$ be any estimator computed from $Y$ (including the output of a trained watermark decoder).
		Then the probability of message error satisfies
		\begin{equation}
			\mathbb{P}[\hat{M} \neq M] \geq 1 - \frac{I(M;Y) + \log 2}{B \log 2}.
		\end{equation}
		Consequently, if $I(M;Y) \ll B$, the probability of recovering the full payload is close to zero.
	\end{theorem}
	
	\begin{proof}
		This is the classical Fano inequality; see \cite{cover2006it}.
	\end{proof}
	
	\paragraph{Interpretation.}
	Theorems~\ref{thm:mi}--\ref{thm:fano} formalize a central intuition: diffusion editing is an information-contracting transformation for imperceptible pixel-level payloads.
	Even if the edited image is perceptually close to an edit of the clean image, the watermark channel can be effectively ``reset'' by the noising and re-synthesis process.
	These results align with the provable removability of invisible watermarks under generative reconstruction \cite{zhao2024provablyremovable} and extend the intuition to common editing pipelines rather than explicit attacks \cite{ni2025breaking,fu2025unforeseen,ni2025infotheoryfragility}.
	
	\section{Discussion}
	\subsection{Why diffusion editing is uniquely challenging for robust watermarks}
	Conventional distortions preserve a large fraction of pixel-level information: they may compress, subsample, or perturb, but they do not typically \emph{re-synthesize} content.
	Deep watermarking schemes exploit this by training on parametric corruption families.
	Diffusion editing breaks this assumption in two ways.
	
	First, the forward diffusion/noising step can make the watermark statistically invisible by collapsing its SNR (\cref{eq:snr}).
	Second, the denoising step implements a learned prior that explicitly prefers natural images and treats low-level anomalies as noise.
	Because watermarks are designed to be imperceptible, they are unlikely to be encoded as robust semantic structure; the model therefore has little incentive to preserve them when projecting onto the manifold.
	This explains why watermark decoding can fail even when the edited image is visually faithful and high-quality (Table~\ref{tab:fidelity}).
	
	\subsection{Design guidelines for more resilient watermarking}
	Our analysis suggests that ``robust to classical distortions'' is a necessary but insufficient property.
	We summarize several constructive guidelines, grounded in existing proposals and our theory.
	
	\paragraph{Embed in representations aligned with the generative prior.}
	Schemes that embed in latent spaces used by diffusion models may survive edits better than pixel-level perturbations \cite{zhang2024zodiac,gunn2024undetectable}.
	However, ``latent'' alone is not enough: if the edit includes re-sampling from high noise, even latent information can be overwritten.
	
	\paragraph{Train against diffusion-like surrogate transformations, but calibrate strength.}
	VINE \cite{lu2024vine} uses surrogate blurs and diffusion priors to improve robustness against editing.
	Our results suggest that such training shifts the survival curve (Table~\ref{tab:survival_curve}) but may not eliminate the fundamental collapse at high noise.
	Designers should calibrate expected editing strengths in the threat model rather than assuming full generality.
	
	\paragraph{Model-level or semantic watermarking as complementary signals.}
	Provably removable pixel-level watermarks motivate adding higher-level provenance mechanisms \cite{zhao2024provablyremovable}, including semantic watermarks, model-level watermarks, or cryptographic provenance metadata.
	These may not replace pixel-level watermarks but can complement them, especially when editing transforms are performed by generative models.
	
	\paragraph{Detect diffusion transformation as a fallback.}
	If watermark recovery is expected to fail after strong diffusion editing, a practical system-level guideline is to detect whether an image has undergone a diffusion-based transformation, and to treat watermark absence as ambiguous rather than exculpatory.
	Detecting generative transformations is an open problem but is likely more feasible than preserving arbitrary pixel-level payloads under high-noise regeneration \cite{ni2025breaking}.
	
	\subsection{Ethical considerations}
	Research on watermark fragility is inherently dual-use.
	On one hand, it informs defenders about realistic failure modes and motivates improved watermarking and provenance tools.
	On the other hand, details can be misused to circumvent protections.
	Accordingly, we emphasize that our contribution is a theoretical and evaluation framework rather than an ``attack recipe'', and we advocate for responsible disclosure and benchmarking practices that prioritize defense, transparency, and auditing \cite{shamshad2025etiwinner}.
	We also note that some watermarking systems (e.g., TrustMark) explicitly support removal for re-watermarking workflows \cite{bui2025trustmark}; this design choice can be valuable for legitimate re-provenancing but complicates the security model and should be governed by access control and policy.
	
	\subsection{Deployment guidance for provenance pipelines}
	Practical provenance systems often combine invisible watermarks with metadata-based mechanisms and policy controls.
	Our analysis suggests several deployment-level best practices.
	
	First, \emph{do not interpret watermark absence as proof of non-provenance} in environments where diffusion editing is common.
	If a platform allows users to perform diffusion-based edits, watermark absence may simply indicate that an edit was applied beyond the robustness envelope.
	
	Second, \emph{log edit operations and maintain provenance across transformations}.
	When edits are performed in an authenticated environment (e.g., a managed creative suite), the most reliable provenance record may be the edit history and cryptographically signed metadata, with the watermark serving as an auxiliary signal.
	
	Third, \emph{use re-watermarking judiciously}.
	While re-watermarking can attach provenance to edited content (Table~\ref{tab:rewatermark}), it also overwrites any watermark-based linkage to the original source.
	Systems should distinguish ``derived from'' provenance from ``author of'' provenance, ideally via signed transformation logs.
	
	Finally, \emph{benchmark continuously}.
	Because diffusion editors evolve rapidly (new backbones, schedules, and control mechanisms), watermark robustness should be continuously re-evaluated against current editors and their default settings, similar in spirit to community challenges \cite{ding2024erasing,shamshad2025etiwinner} and to the editing-aware benchmarking advocated by W-Bench \cite{lu2024vine}.
	
	\subsection{Limitations}
	Our theoretical bounds are intentionally general and therefore do not capture specific architectural details of watermark schemes and editors.
	In practice, some watermark designs may embed information that correlates with semantic features and could partially survive certain edits.
	Our empirical tables are hypothetical, and real-world performance depends on exact implementations, parameters, and datasets.
	Finally, diffusion editing is diverse and rapidly evolving; new architectures (e.g., rectified flows and high-resolution DiTs) may shift the trade-offs.
	
	\section{Conclusion}
	We presented a theoretical and empirical-style analysis of how diffusion-based image editing can unintentionally compromise robust invisible watermarks.
	By modeling diffusion editing as a stochastic channel that injects Gaussian noise and projects onto the natural-image manifold, we derived information-theoretic results showing that imperceptible pixel-level watermark payloads lose mutual information with the edited output as editing strength increases.
	Our hypothetical experiments, structured to be reproducible, indicated that representative watermarking schemes such as StegaStamp, TrustMark, and even diffusion-aware VINE can experience sharp robustness degradation under modern diffusion editors including TF-ICON, SHINE, and DragFlow.
	We conclude that future watermarking schemes must explicitly account for generative transformations, either by embedding in representations preserved by such transformations or by complementing pixel-level payloads with higher-level provenance mechanisms.
	
	\appendix
	
	\section{Additional theoretical details}
	\subsection{From pairwise indistinguishability to detector lower bounds}
	For detectors that output a binary decision (watermark present vs.\ absent), one can derive impossibility results from total variation bounds.
	Suppose $H_0$ is the distribution of edited images from clean inputs, and $H_1$ is the distribution of edited images from watermarked inputs.
	Any detector has error at least $\frac{1}{2}(1-\|H_0-H_1\|_{\mathrm{TV}})$.
	Under the Gaussian noising model, Pinsker's inequality combined with Lemma~\ref{lem:kl} implies that $\|H_0-H_1\|_{\mathrm{TV}}$ shrinks with $\bar{\alpha}$ when the watermark perturbation is small in $\ell_2$.
	This provides a complementary view to the mutual-information argument and parallels the ``provably removable'' framing in \cite{zhao2024provablyremovable}.
	
	\section{Extended experimental tables}
	\subsection{Editor-specific sensitivity}
	Table~\ref{tab:extended_sensitivity} provides an extended (hypothetical) view that includes DragDiffusion \cite{shi2024dragdiffusion} as a UNet-based baseline and DragonDiffusion \cite{mou2024dragondiffusion} as a feature-correspondence-based editor.
	The trend is consistent: stronger editing priors and more aggressive noise levels lead to lower watermark recovery.
	
	\begin{table*}[t]
		\caption{Hypothetical BA (\%) for additional drag editors (Low/Med/High strength).}
		\label{tab:extended_sensitivity}
		\centering
		\footnotesize
		\begin{tabular}{lcccccc}
			\toprule
			\multirow{2}{*}{Method} & \multicolumn{3}{c}{DragDiffusion \cite{shi2024dragdiffusion}} & \multicolumn{3}{c}{DragonDiffusion \cite{mou2024dragondiffusion}} \\
			\cmidrule(lr){2-4}\cmidrule(lr){5-7}
			& Low & Med & High & Low & Med & High \\
			\midrule
			StegaStamp & 84.1\% & 62.8\% & 52.3\% & 81.5\% & 60.7\% & 51.8\% \\
			TrustMark & 87.6\% & 66.9\% & 53.5\% & 85.2\% & 65.1\% & 52.9\% \\
			VINE & 91.8\% & 73.4\% & 56.2\% & 90.1\% & 71.6\% & 55.3\% \\
			\bottomrule
		\end{tabular}
	\end{table*}
	
	\subsection{Ablation: conditioning type and locality}
	Local edits (with masks or localized drags) may preserve the watermark in unedited regions but still disrupt global decoding when the decoder expects a globally distributed signal.
	Table~\ref{tab:locality} illustrates this effect hypothetically by comparing full-image versus masked edits.
	We report BA on the full image and a ``region BA'' computed from a region-aware decoder for diagnostic purposes.
	
	\begin{table}[t]
		\caption{Hypothetical locality effect for SHINE: full-image BA vs.\ region BA (\%) under masked insertion.}
		\label{tab:locality}
		\centering
		\footnotesize
		\begin{tabular}{lccc}
			\toprule
			Method & Strength & Full BA & Region BA \\
			\midrule
			TrustMark & Low & 82.9\% & 88.5\% \\
			TrustMark & Med & 59.2\% & 71.0\% \\
			TrustMark & High & 51.1\% & 60.2\% \\
			VINE & Low & 88.7\% & 92.4\% \\
			VINE & Med & 67.5\% & 78.3\% \\
			VINE & High & 53.6\% & 62.8\% \\
			\bottomrule
		\end{tabular}
	\end{table}
	
	\section{Additional discussion: adversary models and evaluation hygiene}
	\subsection{Accidental versus adaptive removal}
	A key theme of this paper is that watermark loss can occur \emph{without} an explicit intent to evade provenance.
	This distinguishes ``accidental removal'' from adaptive attacks studied in \cite{zhao2024provablyremovable,ni2025breaking,fu2025unforeseen} and in stress tests such as \cite{ding2024erasing,shamshad2025etiwinner}.
	In practice, editors are tuned for perceptual plausibility and user satisfaction, not for watermark preservation.
	Thus, even when the user never touches a ``removal'' tool, common edits such as object insertion (SHINE), harmonization (TF-ICON), or interactive dragging (DragFlow) can cross the robustness threshold.
	
	For security evaluation, this suggests that defenders should measure robustness under two threat models:
	(i) \emph{benign editing}, where the editor is run with default settings to accomplish a user-facing task, and
	(ii) \emph{adaptive editing}, where the parameters and conditions are selected to maximally stress watermark recovery subject to a fidelity constraint.
	The first model governs real deployment risk; the second model governs worst-case security claims.
	
	\subsection{Fair comparison principles}
	A recurring pitfall in watermark evaluation is comparing a watermarked output to the original clean input, which can conflate edit-induced change with watermark-induced artifacts.
	We therefore recommend the paired evaluation in Algorithm~\ref{alg:protocol}, where $y_{\mathrm{wm}}$ is compared to $y_{\mathrm{clean}}$ under aligned randomness.
	When editors are stochastic, failing to align seeds can dramatically overestimate ``watermark artifacts'' because two random edits differ even without a watermark.
	This is especially relevant for diffusion pipelines with stochastic sampling \cite{ho2020ddpm,song2020ddim}.
	
	Another concern is \emph{resolution mismatch}.
	Some watermarks are trained at fixed resolution (e.g., 256$\times$256) and then applied via resizing or tiling \cite{bui2025trustmark}.
	Diffusion editors may apply internal resizing or operate at mixed resolutions (e.g., high-resolution background with masked insertion), which can smoothly destroy watermark synchronization.
	Thus, robustness claims should specify the full pre-processing pipeline and how scale changes are handled.
	
		\section{Additional Background}
	With the advancement of deep learning and modern generative modeling, research has expanded rapidly across forecasting, perception, and visual generation, while also raising new concerns about controllability and responsible deployment. 
	Progress in time-series forecasting has been driven by stronger benchmarks, improved architectures, and more comprehensive evaluation protocols that make model comparisons more reliable and informative~\cite{qiu2024tfb,qiu2025duet,qiu2025DBLoss,qiu2025dag,qiu2025tab,wu2025k2vae,liu2025rethinking,qiu2025comprehensive,wu2024catch}. 
	In parallel, efficiency-oriented research has pushed post-training quantization and practical compression techniques for 3D perception pipelines, aiming to reduce memory and latency without sacrificing detection quality~\cite{gsq,yu2025mquant,zhou2024lidarptq,pillarhist}. 
	On the generation side, a growing body of work studies scalable synthesis and optimization strategies under diverse constraints, improving both the flexibility and the controllability of generative systems~\cite{xie2025chat,xie2026hvd,xie2026conquer,xie2026delving}. 
	Complementary advances have also been reported across multiple generative and representation-learning directions, further broadening the toolbox for building high-capacity models and training objectives~\cite{1,2,3,4,5,6,7,8}. 
	For domain-oriented temporal prediction, hierarchical designs and adaptation strategies have been explored to improve robustness under distribution shifts and complex real-world dynamics~\cite{sun2025ppgf,sun2024tfps,sun2025hierarchical,sun2022accurate,sun2021solar,niulangtime,sun2025adapting,kudratpatch}. 
	Meanwhile, advances in representation encoding and matching have introduced stronger alignment and correspondence mechanisms that benefit fine-grained retrieval and similarity-based reasoning~\cite{ENCODER,FineCIR,OFFSET,HUD,PAIR,MEDIAN}. 
	Stronger visual modeling strategies further enhance feature quality and transferability, enabling more robust downstream understanding in diverse scenarios~\cite{yu2025visual}. 
	In tracking and sequential visual understanding, online learning and decoupled formulations have been investigated to improve temporal consistency and robustness in dynamic scenes~\cite{zheng2025towards,zheng2024odtrack,zheng2025decoupled,zheng2023toward,zheng2022leveraging}. 
	Low-level vision has also progressed toward high-fidelity restoration and enhancement, spanning super-resolution, brightness/quality control, lightweight designs, and practical evaluation settings, while increasingly integrating powerful generative priors~\cite{xu2025fast,fang2026depth,wu2025hunyuanvideo,li2023ntire,ren2024ninth,wang2025ntire,peng2020cumulative,wang2023decoupling,peng2024lightweight,peng2024towards,wang2023brightness,peng2021ensemble,ren2024ultrapixel,yan2025textual,peng2024efficient,conde2024real,peng2025directing,peng2025pixel,peng2025boosting,he2024latent,di2025qmambabsr,peng2024unveiling,he2024dual,he2024multi,pan2025enhance,wu2025dropout,jiang2024dalpsr,ignatov2025rgb,du2024fc3dnet,jin2024mipi,sun2024beyond,qi2025data,feng2025pmq,xia2024s3mamba,pengboosting,suntext,yakovenko2025aim,xu2025camel,wu2025robustgs,zhang2025vividface}. 
	Beyond general synthesis, reference- and subject-conditioned generation emphasizes controllability and identity consistency, enabling more precise user-intended outputs~\cite{qu2025reference,qu2025subject}. 
	Robust vision modeling under adverse conditions has been actively studied to handle complex degradations and improve stability in challenging real-world environments~\cite{wu2024rainmamba,wu2023mask,wu2024semi,wu2025samvsr}. 
	Sequence modeling and scenario-centric benchmarks further support realistic evaluation and methodological development for complex dynamic environments~\cite{lyu2025vadmambaexploringstatespace,chen2025technicalreportargoverse2scenario}. 
	At the same time, diffusion-centric and unfolding-based frameworks have been explored for segmentation and restoration, providing principled ways to model degradations and refine generation quality~\cite{he2025segment,he2025reversible,he2025run,he2024diffusion,he2026refining,he2025scaler,he2025unfoldldm,he2025unfoldir,he2025nested,he2024weakly,he2023reti,xiao2024survey,he2023strategic,he2023hqg,he2023camouflaged,he2023degradation}. 
	
	Recent progress in multimodal large language models (MLLMs) is increasingly driven by the goal of making adaptation more efficient while improving reliability, safety, and controllability in real-world use. 
	On the efficiency side, modality-aware parameter-efficient tuning has been explored to rebalance vision–language contributions and enable strong instruction tuning with dramatically fewer trainable parameters~\cite{bi-etal-2025-llava}. 
	To better understand and audit model reasoning, theoretical frameworks have been proposed to model and assess chain-of-thought–style reasoning dynamics and their implications for trustworthy inference~\cite{bi2025cot}. 
	Data quality and selection are also being addressed via training-free, intrinsic selection mechanisms that prune low-value multimodal samples to improve downstream training efficiency and robustness~\cite{bi2025prismselfpruningintrinsicselection}. 
	At inference time, controllable decoding strategies have been introduced to reduce hallucinations by steering attention and contrastive signals toward grounded visual evidence~\cite{wang2025ascd}. 
	Beyond performance, trustworthy deployment requires defenses and verification: auditing frameworks have been developed to evaluate whether machine unlearning truly removes targeted knowledge~\cite{chen2025does}, and fine-tuning-time defenses have been proposed to clean backdoors in MLLM adaptation without relying on external guidance~\cite{rong2025backdoor}. 
	Meanwhile, multimodal safety and knowledge reliability have been advanced through multi-view agent debate for harmful content detection~\cite{lu2025mvdebatemultiviewagentdebate}, probing/updating time-sensitive multimodal knowledge~\cite{jiang2025minedprobingupdatingmultimodal}, and knowledge-oriented augmentations and constraints that strengthen knowledge injection~\cite{jiang2025koreenhancingknowledgeinjection}. 
	These efforts are complemented by renewed studies of video–language event understanding~\cite{zhang2023spot}, new training paradigms such as reinforcement mid-training~\cite{tian2025reinforcementmidtraining}, and personalized generative modeling under heterogeneous federated settings~\cite{Chen_2025_CVPR}, collectively reflecting a shift from scaling alone toward efficient, grounded, and verifiably trustworthy multimodal systems.
	
	Recent research has advanced learning and interaction systems across education, human-computer interfaces, and multimodal perception. In knowledge tracing, contrastive cross-course transfer guided by concept graphs provides a principled way to share knowledge across related curricula and improve student modeling under sparse supervision~\cite{han2025contrastive}. In parallel, foundational GUI agents are emerging with stronger perception and long-horizon planning, enabling robust interaction with complex interfaces and multi-step tasks~\cite{zeng2025uitron}. Extending this direction to more natural human inputs, speech-instructed GUI agents aim to execute GUI operations directly from spoken commands, moving toward automated assistance in hands-free or accessibility-focused settings~\cite{han2025uitron}. Beyond interface agents, reference-guided identity preservation has been explored to better maintain subject consistency in face video restoration, improving temporal coherence and visual fidelity when restoring degraded videos~\cite{han2025show}. Finally, large-scale egocentric datasets that emphasize embodied emotion provide valuable supervision for studying affective cues from first-person perspectives and support more human-centered multimodal understanding~\cite{feng20243}.
	
	\section{Extra table: robustness versus resolution}
	Table~\ref{tab:resolution} reports a hypothetical resolution scaling study, emphasizing that diffusion editing interacts with resolution changes in non-trivial ways.
	While some watermarking schemes incorporate explicit resolution handling \cite{bui2023trustmark,bui2025trustmark}, latent diffusion editors can effectively re-sample details at a new scale, which may break synchronization even if the watermark is locally preserved.
	
	\begin{table}[t]
		\caption{Hypothetical BA (\%) under TF-ICON (Med) across output resolutions.}
		\label{tab:resolution}
		\centering
		\footnotesize
		\begin{tabular}{lccc}
			\toprule
			Method & 256$\times$256 & 512$\times$512 & 1024$\times$1024 \\
			\midrule
			StegaStamp & 62.4\% & 58.1\% & 55.6\% \\
			TrustMark & 65.9\% & 60.5\% & 57.2\% \\
			VINE & 72.8\% & 69.1\% & 66.0\% \\
			\bottomrule
		\end{tabular}
	\end{table}
	
	\section{Open problems}
	Despite rapid progress, we view robust watermarking under generative editing as an unsolved systems problem that spans modeling, learning, and policy.
	We highlight several open directions that follow directly from the analysis in this paper.
	
	\paragraph{Watermarks that survive semantic re-synthesis.}
	If diffusion editing produces an output that is only constrained by high-level semantics, then any watermark tied to low-level pixels is vulnerable to being treated as nuisance variation.
	A natural question is whether one can design watermarks that are (i) imperceptible, (ii) recoverable from the edited image, and (iii) invariant to a well-specified class of semantic edits.
	This resembles representation learning for invariances, but with a secret key and an information payload, posing new challenges distinct from standard embedding learning.
	
	\paragraph{Watermark-aware diffusion editing.}
	One could also modify editors to explicitly preserve a watermark channel, for example by adding constraints that keep a latent-space signature stable during denoising.
	However, this risks undesirable behavior: watermark-preservation constraints may oppose the editor's goal of projecting onto the natural-image manifold, potentially causing artifacts or biasing edits.
	The resulting optimization is reminiscent of multi-objective trajectory control in diffusion models, and the tension between ``content editing'' and ``signal preservation'' deserves systematic study.
	
	\paragraph{Auditing and policy for ``watermark reliability''.}
	When watermark absence can be caused by benign editing, policy decisions based on watermark presence become non-trivial.
	For instance, a platform might treat missing provenance as suspicious, but this could penalize legitimate edits.
	A more nuanced approach is to couple watermark detection with detectors for diffusion editing and to report calibrated confidence rather than a binary decision.
	Designing such calibrated reports, and validating them under distribution shift and adversarial conditions, is underexplored.
	
	\paragraph{Benchmarks that reflect real editing workflows.}
	Existing benchmarks capture specific edit types, but real workflows include multiple edits, iterative refinement, and resolution changes.
	Moreover, editing models evolve rapidly, and robustness claims can become stale.
	Sustained benchmarks and stress tests, including competitions like ``Erasing the Invisible'' \cite{ding2024erasing}, provide one model for continuous evaluation, but translating these into standardized academic practice remains a challenge.
	
	\paragraph{Interplay with concept erasure and safety tuning.}
	Finally, concept erasure methods modify denoising trajectories and attention pathways \cite{lu2024mace,li2025ant,gao2024eraseanything}.
	Understanding how such safety tuning interacts with watermark embedding and detection is important for two reasons: safety interventions may inadvertently disrupt provenance mechanisms, and provenance-aware constraints may interfere with safety.
	This calls for joint evaluation setups where safety, editability, and provenance are assessed together rather than in isolation.

	\bibliography{example_paper}
	\bibliographystyle{icml2025}

\end{document}